%% file: main.tex
\documentclass[a4paper,margin=1in]{article}
\usepackage{geometry}
\usepackage{hyperref}

\usepackage{mathtools}
\usepackage{amsfonts}
\usepackage{amsthm}
\usepackage[usenames,dvipsnames]{xcolor}
\usepackage[toc,page]{appendix}
\usepackage{tikz}
\usepackage{esvect}
\usepackage{xcolor}
\usetikzlibrary{patterns}
\usetikzlibrary{decorations.pathreplacing}

\newtheorem{observation}{Observation}
\newtheorem{definition}{Definition}
\newtheorem{lemma}{Lemma}
\newtheorem{theorem}{Theorem}

\newcommand{\floor}[1]{\left\lfloor #1 \right\rfloor}
\newcommand{\ceil}[1]{\left\lceil #1 \right\rceil}
\newcommand{\dotpr}[2]{\left\langle #1 , #2 \right\rangle}
\newcommand{\abs}[1]{| #1 |}
\newcommand{\Bigabs}[1]{\Big\lvert #1 \Big\rvert}
\newcommand{\biggabs}[1]{\bigg\lvert #1 \bigg\rvert}
\newcommand{\dist}[1]{\lVert #1 \rVert}
\newcommand{\prob}[1]{\mathbb{P}\left[ #1 \right]}
\newcommand{\Ex}{\mathbb{E}}
\newcommand{\Exp}[1]{\exp\left( #1 \right)}

\newcommand{\poly}[1]{\mathrm{poly}( #1 )}
\newcommand{\RCNN}{(R,c)-\textsf{NN} }
\newcommand{\Oc}{\mathcal{O}}
\newcommand{\Rspace}{\mathbb{R}}
\newcommand{\Rdspace}{\mathbb{R}^d}

\newcommand{\Nspace}{\mathbb{N}}
\newcommand{\D}{\mathcal{D}}

\newcommand{\pfp}{\mbox{$\textsf{p}_{\textsf{fp}}$}}
\newcommand{\rhop}{\mbox{$\rho_p$}}
\renewcommand*{\Vec}[1]{\vv{#1}}

\newcommand{\footremember}[2]{%
    \footnote{#2}
    \newcounter{#1}
    \setcounter{#1}{\value{footnote}}%
}
\newcommand{\footrecall}[1]{%
    \footnotemark[\value{#1}]%
}

\title{Locality-Sensitive Hashing without False Negatives for $l_p$}

\author{Andrzej Pacuk\footremember{MIMUW}{Institute of Informatics, University
of Warsaw, Poland} \and Piotr Sankowski\footrecall{MIMUW} \and Karol Wegrzycki\footrecall{MIMUW} \and Piotr
Wygocki\footrecall{MIMUW}}
\date{\texttt{[apacuk,sank,k.wegrzycki,wygos]@mimuw.edu.pl}}

\begin{document}

\maketitle


\input{abstract}
\input{introduction}

\input{related_work}
\input{basic_construction}

\input{complexity}

\input{conclusion}

\bibliography{bib}

\include{appendix}

\end{document}

%% file: abstract.tex
\begin{abstract}

In this paper, we show a construction of locality-sensitive hash functions
without false negatives, i.e., which ensure collision for every pair of points
within a given radius $R$ in $d$ dimensional space equipped with $l_p$ norm when
$p \in [1,\infty]$. Furthermore, we show how to use these hash functions to
solve the $c$-approximate nearest neighbor search problem without false
negatives. Namely, if there is a point at distance $R$, we will certainly report
it and points at distance greater than $cR$ will not be reported for
$c=\Omega(\sqrt{d},d^{1-\frac{1}{p}})$.
The constructed algorithms work:
\begin{itemize}
    \item with preprocessing time $\Oc(n \log(n))$ and sublinear expected query
        time,
    \item with preprocessing time $\Oc(\mathrm{poly}(n))$ and expected query \newline time $\Oc(\log(n))$.
\end{itemize}
Our paper reports progress on answering the open problem presented by
Pagh~\cite{Pagh15}, who considered the nearest neighbor search without false
negatives for the Hamming distance.

\end{abstract}

%% file: introduction.tex
\section{Introduction}

The \textit{Nearest Neighbor} problem is of major importance to a variety of
applications in machine learning and pattern recognition. Ordinarily,
points are embedded in $\Rdspace$, and distance metrics usually
measure similarity between points. Our task is the following:  given a preprocessed set of points
$S \subset \Rdspace$ and a query point $q \in \Rdspace$, find the point $v \in S$,
with the minimal distance to $q$. Unfortunately, the existence of an efficient
algorithm (i.e., whose query and preprocessing time would not depend
exponentially on $d$), would disprove the strong exponential time hypothesis~\cite{Pagh15,Williams2005}.
Due to this fact, we consider the \emph{$c$-approximate nearest neighbor}
problem: given a distance $R$, a query point $q$ and a constant $c > 1$, we need to find a
point within distance $c R$ from point $q$~\cite{Datar04locality-sensitivehashing}.
This point is called a \emph{$c R$-near neighbor} of $q$.

\begin{definition}
    Point $v$ is an \textit{$r$-near neighbor} of $q$ in metric $\mathcal{M}$
    iff $\mathcal{M}(q,v) \le r$.
\end{definition}

One of the most interesting methods for solving the $c$-approximate nearest neighbor problem in high-dimensional space is 
\textit{locality-sensitive hashing} (LSH). The algorithm offers a sub-linear query time and
a sub-quadratic space complexity. The rudimentary component on which LSH method relies
is \textit{locality-sensitive hashing function}. Intuitively, a hash
function is \textit{locality-sensitive} if the probability of collision is much higher for
``nearby'' points than for ``far apart'' ones. More formally:


\begin{definition}
    A family $H = \{h: S \rightarrow U \}$ is called $(r,c,p_1,p_2)$-sensitive for
    distance $\mathcal{D}$ and induced ball $\mathcal{B}(q,r)= \{v : \mathcal{D}(q,v)  < r \}$, if for any $v,q \in S$:
    \begin{itemize}
        \item if $v \in \mathcal{B}(q,r)$ then $\prob{h(q) = h(v)} \ge p_1$,
        \item if $v \notin \mathcal{B}(q,c r)$ then $\prob{h(q) = h(v)} \le p_2$.
    \end{itemize}
    For $p_1 > p_2$ and $c>1$.
\end{definition}

Indyk and Motwani~\cite{motwani} considered
randomized $c$-approximate $R$-near neighbor~(Definition~\ref{rnear}).

\begin{definition}[The randomized $c$-approximate $R$-near neighbor or \RCNN]
    Given a set of points in a $P \subset \Rdspace$ and parameters $R>0$, $\delta > 0$.
    Construct a data structure $D$ such that for any query point $q$, if there
    exists a R-near neighbor of $q$ in $P$, $D$ reports some $c R$-near neighbor
    of $q$ in $P$ with probability $1-\delta$.

    \label{rnear}
\end{definition}

In this paper, we study guarantees for LSH based \RCNN such that for each query
point $q$, every close enough point $\dist{x-q}_p < R$ will be certainly returned, i.e., there are no false 
negatives.\footnote{$\dist{\cdot}_p$ denotes the standard $l_p$ norm for fixed $p$.}
In other words, given a set $S$ of size $n$ and a query point $q$,
the result is a set $P \subseteq S$ such that:
\begin{displaymath}
 \{x : \dist{x-q}_p  < r\} \subseteq P \subseteq \{x : \dist{x-q}_p  \le c r \}.
\end{displaymath}
Moreover, for each distant point ($\dist{x-q}_p > c R$), the probability of being returned is
bounded by \pfp  -- probability of false positives.
In~\cite{Pagh15} this type of LSH is called \textit{LSH without false
negatives}. The fact that the probability of false negatives is $0$ is our main
improvement over Indyk and Motwani algorithm~\cite{motwani}. Furthermore, Indyk
and Motwani showed that $p$-stable distributions (where $p \in (0,2]$) are
$(r,c,p_1,p_2)$-sensitive for $l_p$. We generalized their results on any
 distribution with mean $0$, bounded second and fourth moment and any $p
\in[1,\infty]$ (see Lemma~\ref{big}, for rigorous definitions).
Finally, certain distributions from this abundant class guarantee that points
within given radius will always be returned (see
Figure~\ref{fig:circle}).
Unfortunately, our results come with a price, namely $c \ge \max\{\sqrt{d},
d^{1-1/p}\}$.

\input{circle_picture}

%% file: circle_picture.tex
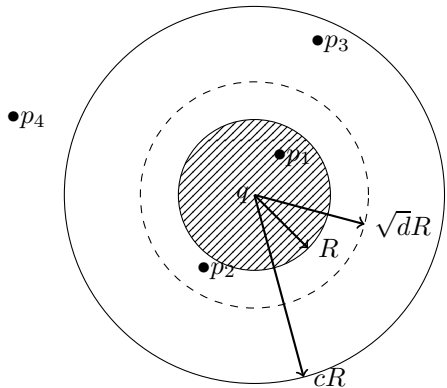
\begin{figure*}[ht!]

\centering

\begin{tikzpicture}[scale = 0.5]

\draw[pattern=north east lines] (0,0) circle (2cm);
\draw[dashed] (0,0) circle (3cm);
\draw(0,0) circle (5cm);

\begin{scope}[rotate=315]
    \draw[thick,->] (0,0) -- (2,0) node[anchor=west] {$R$};
\end{scope}
\begin{scope}[rotate=345]
    \draw[thick,->] (0,0) -- (3,0) node[anchor=west] {$\sqrt{d} R$};
\end{scope}

\begin{scope}[rotate=285]
    \draw[thick,->] (0,0) -- (5,0) node[anchor=west] {$c R$};
\end{scope}

\node[align=left] at (0,0)   {$q$ \;\;};
\node at (1,1)   {\textbullet $p_1$};
\node at (2,4)   {\textbullet $p_3$};
\node at (-1,-2) {\textbullet $p_2$};
\node at (-6,2) {\textbullet $p_4$};

\end{tikzpicture}
\caption{The presented algorithms guarantee that points in the dashed area ($p_1$) will be reported as neighbors.
    Points within the dotted circle ($p_2$) will be reported as neighbor with high probability.
    Points ($p_3$) within a distance $cR$  might be reported, but not necessarily.
    Points ($p_4$) outside circle $cR$ cannot be reported. The schema picture
    presents an example for the euclidean distance ($p=2$).
    }
\label{fig:circle}
\end{figure*}

%% file: related_work.tex
\section{Related Work}
\subsection{Nearest neighbor in high dimensions}
Most common techniques for solving the approximate nearest neighbor search, such as
the spatial indexes or k-d trees~\cite{Bentley:1990:KDT:98524.98564} are designed to work well
for the relatively small number of dimensions. The query time for k-d trees is
$\Oc(n^{1-\frac{1}{d}})$ and when the number of dimensions increases the complexity basically
converges to $\Oc(n)$. For interval trees, query time $\Oc(\log^d{n})$
depends exponentially on the number of dimensions. The major
breakthrough was the result of Indyk and Motwani~\cite{motwani}. Their algorithm
has expected complexity of $\Oc(d n^{\frac{1}{c}})$ for any approximation
constant $c > 1$ and the complexity is tight for any metric
$l_p$ (where $p \in (0,2]$). Indyk and Motwani introduced the following LSH functions:
\begin{displaymath}
    h(v) = \floor{\frac{\dotpr{a}{v} + b}{r}},
\end{displaymath}
where $a$ is the $d$-dimensional vector of independent random variables from
a $p$-stable distribution and $b$ is a real number chosen uniformly from the range
$[0,r]$.

Our algorithm is based on similar functions and we prove compelling results for
more general family of distributions (we show bounds for any
distribution with a bounded variance and an expected value equal to $0$). Furthermore, our algorithm is correct
for any $p \in [1, \infty]$. Indyk and Motwani`s LSH algorithm
was showed to be optimal for $l_1$ metric. Subsequently, Andoni~et~al.~\cite{DBLP:journals/cacm/AndoniI08} showed near optimal results for $l_2$.
Recently, data dependant techniques have been used to further improve
LSH by Andoni and Razenshteyn~\cite{data-depended-hashing}. However, the constant $\rho$
in a query time $\Oc(n^\rho)$ remains:
\begin{displaymath}
    \rho = \frac{\log{p_1}}{\log{p_2}}.
\end{displaymath}
When a formal guarantee that $p_1 = 1$ is needed their algorithm does not apply.

\subsection{LSH without false negatives}
Recently, Pagh~\cite{Pagh15} presented a novel approach to nearest neighbor
search in Hamming space. He showed the construction of an efficient
locality-sensitive hash function family that guarantees collision for any
close points. Moreover, Pagh showed that bounds of his algorithm for $c r =
\log{n}/k$ (where $k \in \Nspace$) essentially match bounds of Indyk and
Motwani (differ by at most factor $\ln{4}$ in the exponent).
More precisely, he showed that the problem of false negatives can be avoided
in the Hamming space at some cost in efficiency.
He proved bounds for general values of $c$. This paper is an answer to his open
problem: whether is it possible to get similar results for other distance
measures (e.g., $l_1$ or $l_2$).

Pagh introduced the concept of an \textit{$r$-covering} family of hash function:
\begin{definition}
    For $A \subseteq \{0,1\}^d$, the Hamming projection family
    $\mathcal{H}_\mathcal{A}$ is $r$-covering if for every $x \in \{0,1\}^d$ with
    $\dist{x}_H \leq r$, there exist $h \in \mathcal{H}_\mathcal{A}$ such that $h(x)
    =  \textbf{0}$.
\end{definition}
Then, he presented a fast method of generating such an $r$-covering family.
Finally, he showed that the expected number of false positives is bounded by
$2^{r+1-\dist{x-y}_H}$.

%% file: basic_construction.tex
\section{Basic Construction} \label{lsh_sec}

We will consider the $l_p$ metric for $p \in [1,\infty]$ and $n$ fixed points in $\Rdspace$ space.
Let $v$ be a $d$-dimensional vector of independent random variables drawn from distribution $\D$. We define a function $h_p$ as:
\begin{displaymath}
    h_p(x) = \floor{
    {\frac{ \dotpr{x}{v} }{r \rhop}}
    },
\end{displaymath}
where $ \dotpr{}{}$ is a standard inner product and $\rhop = d^{1-\frac{1}{p}}$.
The scaling factor $\rhop$ is chosen so that: $\dist{z}_1 \leq \rhop \dist{z}_p$.
The rudimentary distinction between the hash function $h_p$ and LSH is that we
consider two hashes equal when they differ at most by one.
In Indyk and Motwani~\cite{motwani} version of LSH, there were merely
probabilistic guarantees, and close points (say $0.99$ and $1.01$) could
be returned in different buckets with small probability. Since our motivation is
to find all close points with absolute certainty, we need to check the adjacent buckets as
well.

First, observe that for given points, the probability of choosing a hash
function that will classify them as equal is bounded on both sides as given by the following observations.
The proofs of these observations are in Appendices \ref{achar} and \ref{achar2}.

\begin{observation}[Upper bound on the probability of point equivalence]\label{char}
    \begin{displaymath}
    \prob{\abs{ h_p(x) - h_p(y) } \le 1}
        \le
            \prob{
                    \abs{ \dotpr{x-y}{v} }
                <
                    2 \rhop r
            }
    .
    \end{displaymath}
\end{observation}

\begin{observation}[Lower bound on the probability of point equivalence]\label{char2}
    \begin{displaymath}
            \prob{|h_p(x) - h_p(y)| \le 1}
        \ge
            \prob{
                | \dotpr{x-y}{v} | < \rhop r
            }
    .
    \end{displaymath}
\end{observation}

Interestingly, using the aforementioned observations
we can configure a distribution $\D$ so the close points must end up in the same
or adjacent bucket. 
\begin{observation}[Close points have close hashes] \label{small}
For distribution $\D$ such that every $v_i \sim \D$: $-1 \le v_i \le 1$ and
for $x,y \in \Rdspace$, if $\dist{x-y}_p < r$ then $\forall_{h_p} |h_p(x) - h_p(y)| \le 1$.
\end{observation}
\begin{proof}
   We know that $\dist{z}_1 \le \rhop \dist{z}_p$ and $\abs{v_i} \le 1$ (because
   $v_i$ is drawn from bounded distribution $\D$), so:
\begin{eqnarray*}
    \rhop \dist{x-y}_p \ge \dist{x-y}_1
    = \sum_i \abs{x_i - y_i}
    \ge \sum_i \abs{v_i(x_i - y_i)}
    &\ge \Bigabs{\sum_i v_i (x_i-y_i)}\\
    &= \abs{ \dotpr{x-y}{v} }
    .
\end{eqnarray*}
Now, when points are close in $l_p$:
\begin{eqnarray*}
    \dist{x-y}_p < r \iff \rhop \dist{x-y}_p < \rhop r
    \implies \abs{ \dotpr{x-y}{v} } < \rhop r
    .
\end{eqnarray*}
Next, by Observation~\ref{char2}:
\begin{displaymath}
    1 = \prob{ \abs{ \dotpr{x-y}{v} } < \rhop r}
    \le \prob{|h_p(x) - h_p(y)| \le 1}
    .
\end{displaymath}
Hence, the points will inevitably hash into the same or adjacent buckets.
\end{proof}

Now we will introduce the inequality that will help to bound the probability of false positives.

\begin{observation}[Inequality of norms in $l_p$]\label{observation-max}

    Recall that $\rhop = d^{1-\frac{1}{p}}$. For every $ z \in \Rdspace$ and $p \in [
    1,\infty ]$:

    \begin{displaymath}
        \dist{z}_2 \ge \frac{\rhop}{\max \{d^\frac{1}{2}, d^{1 - \frac{1}{p} } \} } \dist{z}_p
        .
    \end{displaymath}

\end{observation}
This technical observation is proven in Appendix \ref{amax}.

The major question arises: what is the probability of false positives? In
contrast to the Indyk and Motwani~\cite{motwani}, we cannot use $p$-stable
distributions because these distributions are not bounded. We will present the proof
for a different class of functions.

\begin{lemma}[The probability of false positives for general distribution]\label{big}

Let $\D$ be a random variable such that $\Ex(\D) = 0$, $\Ex(\D^2) = \alpha^2$,
$\Ex(\D^4) \le 3 \alpha^4$ (for any $\alpha \in \Rspace^{+}$).
Define constant
$\tau_1 = \frac{2}{\alpha} \max \{ d^{\frac{1}{2}}, d^{1-\frac{1}{p}}\}$.

When $\dist{x-y}_p > c r$, $x,y \in \Rdspace$ and $c > \tau_1$ then:

\begin{displaymath}
    \pfp_1 = \prob{|h_p(x) - h_p(y)| \le 1} < 1-\frac{\left(1-\frac{\tau_1^2}{c^2}\right)^2}{3}
    ,
\end{displaymath}

for every metric $l_p$, where $p \in [1,\infty]$ ($\pfp_1$ is the probability of
false positive).

\end{lemma}

\begin{proof}
    By Observation~\ref{observation-max}:
    \begin{displaymath}
        \dist{z}_2 \ge \frac{2 \dist{z}_p}{\alpha \tau_1} \rhop
    \end{displaymath}
    Subsequently, let $z=x-y$ and define a random variable $X = \dotpr{z}{v}$.
    Therefore:
    \begin{displaymath}
    \Ex(X^2) = \alpha^2 \dist{z}_2^2 \ge (\frac{2 \dist{z}_p}{\tau_1} \rhop)^2 >
    (2 r \rhop \frac{c}{\tau_1})^2.
    \end{displaymath}
    Because $\frac{c}{\tau_1} > 1$ we have $\theta = \frac{(2 r
    \rhop)^2}{\Ex{X^2}} < 1$. Variable $\theta$ and a random variable $X^2>0$
    satisfy Paley-Zygmunt inequality (analogously to~\cite{KHIN_IN}):
    \begin{eqnarray*}
        \prob{|h_p(x) - h_p(y)| > 1}
        & \ge & \prob{| \dotpr{z}{v} | \ge 2 r \rhop} \ge  \prob{X^2 > (2 r \rhop)^2} \\
        & \ge & \left(1-\frac{(2 r \rhop)^2}{\Ex(X^2)}\right)^2 \frac{\Ex(X^2)^2}{\Ex(X^4)}
        .
    \end{eqnarray*}
    Eventually, we assumed that $\Ex(X^4) \le 3 (\alpha \dist{z}_2)^4$:
    \begin{displaymath}
        \prob{|h_p(x) - h_p(y)| > 1}
    \ge \frac{\left(1-\frac{(2 r \rhop)^2}{\Ex(X^2)}\right)^2}{3}
        > \frac{\left(1-\frac{\tau_1^2}{c^2}\right)^2}{3}
        .
    \end{displaymath}
\end{proof}

Simple example of a distribution that satisfies both Observation~\ref{small} and
Lemma~\ref{big} is a uniform distribution on $(-1, 1)$ with a standard deviation $\alpha$ equal
to $\sqrt{\frac{1}{3}}$. Another example of such distribution is a discrete distribution with uniform values
$\{-1,1\}$.
As it turns out, Lemma~\ref{big01} shows that the discrete distribution leads to even better bounds.

\begin{lemma}[Probability of false positives for the discrete distribution]\label{big01}
Let $\D$ be a random variable such that $ \prob{\D=\pm 1} = \frac{1}{2}$.
Define constant
$\tau_2 = \sqrt{8}\max \{d^{\frac{1}{2}}, d^{1-\frac{1}{p}}\}$.
Then for every $p \in [1,\infty]$,
$x,y \in \Rdspace$ and $c > \tau_2$ such that $\dist{x-y}_p > cr$, it holds:

\begin{displaymath}
    \pfp_2 = \prob{|h_p(x) - h_p(y)| \le 1} < 1-\frac{(1-\frac{\tau_2}{c})^2}{2}
    .
\end{displaymath}
\end{lemma}
\begin{proof}
    Because of Observation~\ref{observation-max} we have the inequality:
    $$
        \dist{z}_2 \ge \sqrt{8} \frac{\dist{z}_p}{\tau_2} \rhop
        .
    $$
    Let $z=x-y$ and $X = \dotpr{z}{v} $, be a random variable.
    Then:
    $$
        \prob{|h_p(x) - h_p(y)| > 1} \ge
        \prob{|X| > 2 r \rhop}
        .
    $$
    \textit{Khintchine} inequality~\cite{Haagerup1981} states
    $\Ex |X| \ge \frac{\dist{z}_2}{\sqrt{2}}$, so:
    $$
        \Ex(|X|) \ge \frac{\dist{z}_2}{\sqrt{2}} \ge
        \frac{2 \rhop \dist{z}_p}{\tau_2} > 2 r \rhop \frac{c}{\tau_2}
        .
    $$
    Note that, a random variable $|X|$ and $\theta = \frac{2 r \rhop}{\Ex(|X|)} <
    1$, satisfy the Paley-Zygmunt inequality (because $\frac{c}{\tau_2} > 1$), though:
    \begin{align*}
        \prob{h_p(x) - h_p(y)| > 1}
        \ge \left( 1-\frac{2 r \rhop}{\Ex(|X|)} \right)^2 \frac{\Ex(|X|)^2}{\Ex(|X|^2)}  \\
         >   \left(1-\frac{2 r \rhop}{2 r \rhop\frac{c}{\tau_2}}\right)^2
             \frac{1}{2} 
         =   \frac{\left(1-\frac{\tau_2}{c}\right)^2}{2}
         .
    \end{align*}
\end{proof}

Altogether, in this section we have introduced a family of hash functions $h_p$ which:
\begin{itemize}
    \item guarantees that, with an absolute certainty, points within the distance
        $R$ will be mapped to the same or adjacent buckets (see Observation~\ref{small}),
    \item maps ``far away'' points to the non-adjacent hashes with
        high probability (Lemma~\ref{big} and Lemma~\ref{big01}).
\end{itemize}
These properties will enable us to construct an efficient algorithm for solving
the $c$-approximate nearest neighbor search problem without false negatives.

\subsection{Tightness of bounds}\label{tight}

We showed that for two distant points
$x, y : \dist{x-y}_p > cr$, the probability of a collision is small when
$c = \max\{\rho_p, \sqrt{d}\}$. The natural question arises: Can we bound the
probability of a collision for points $\dist{x-y}_p > c'r$ for some $c' < c$?

We will show that such $c'$ does not exist, i.e., there always exists $\tilde{x}$
such that $\dist{\tilde{x}}_p$ will be arbitrarily close to $cr$,
so $\tilde{x}$ and $\Vec{0}$ will end up in the same or adjacent bucket with high probability.
More formally, for any $p \in [1,\infty]$, for
$h_p(x) = \floor{ \frac{ \dotpr{x}{v} }{r \rhop}}$, where coordinates of
$d$-dimensional vector $v$ are
random variables $v_i$, such that $-1 \le v_i \le 1$ with $\Ex(v_i)=0$. We will show that there
always exists $\tilde{x}$ such that
$\dist{\tilde{x}}_p \approx r \max\{\rho_p, \sqrt{d}\}$ and $|h_p(\tilde{x}) -
h_p(\Vec{0})|  \le 1$ with high probability.


For $p \ge 2$ denote $x_0 = (r\rho_p-\epsilon,0, 0, \ldots ,0)$. We
have $\dist{x_0 -\Vec{0}}_p = r\rho_p - \epsilon$ and:
\begin{displaymath}
    |h_p(x_0) - h_p(\Vec{0})| = \biggabs{\floor{\frac{r\rho_p-\epsilon}{r\rho_p} \cdot v_1} - 0} \le 1.
\end{displaymath}

For $p \in [1,2)$, denote $x_1 = rd^{-\frac{1}{p}+\frac{1}{2}-\epsilon}\Vec{1}$.
We have $\dist{x_1}_p = rd^{\frac{1}{2}-\epsilon}$ and by applying
Observation~\ref{char2} for complementary probabilities:

\begin{eqnarray*}
    \prob{|h_p(x_1)-h_p(\Vec{0})| > 1} &\le& \prob{|\dotpr{x_1}{v}| \ge \rhop r}
    =
    \prob{|\dotpr{\Vec{1}}{v}| \ge d^{\frac{1}{2}+\epsilon}}\\
    &=& \prob{\biggabs{\frac{\sum_{i=1}^d v_i}{d}} \ge d^{-\frac{1}{2}+\epsilon}}
    \le 2\cdot \Exp{\frac{-d^{2\epsilon}}{2}}.
\end{eqnarray*}

The last inequality follows from Hoeffding~\cite{Hoeffding:1963} (see Appendix~\ref{hoef} for technical details).

So the aforementioned probability for $p\in [1,2)$ is bounded by an expression
exponential in $d^{2 \epsilon}$. Even if we would concatenate $k$ random hash
functions (see proof of Theorem~\ref{t1} for more details), the chance of
collision would be at least $(1 - 2e^{\frac{-d^{2\epsilon}}{2}})^k$. To bound
this probability, the number $k$ needs to be at least $\Theta{(e^{\frac{d^{2
\epsilon}}{2}})}$.
The probability bounds do not work for $\epsilon$ arbitrary
close to $0$:
we proved that introduced hash functions for $c=d^{1/2-\epsilon}$ do
not work (may give false positives).\footnote{However, one may try to obtain
tighter bound (e.g., $c = d^{1/2}/\log(d)$) or show that for every $\epsilon >
0$, the approximation factor $c=d^{1/2}-\epsilon$ does not work.}

Hence, to obtain a significantly better approximation factor $c$, one
must introduce a completely new family of hash functions.

%% file: complexity.tex
\section{The algorithm}\label{algorithm}

In this section, we apply the LSH family introduced in Section~\ref{lsh_sec} to
construct an \textit{$c$-approximate} algorithm without false negatives.
To begin with, we will define a general algorithm that will satisfy our
conditions. Subsequently, we will show that complexity of the query is
sublinear, and it depends linearly on the number of dimensions.

\begin{theorem}\label{t1}
 For any $c>\tau$ and the number of iterations $k \ge 0$, there exists a
 \textit{$c$-approximate} nearest neighbor algorithm without false negatives for
 $l_p$, where $p \in [1,\infty]$:
 \begin{itemize}
    \item Preprocessing time: $\Oc(n (k d + 3^k))$,
    \item Memory usage: $\Oc(n 3^k)$,
    \item Expected query time: $\Oc(d (|P| + k + n \pfp^k))$.
 \end{itemize}
 Where $|P|$ is the size of the result and $\pfp$ is the upper bound of probability of false
 positives (note that $\pfp$ depends on a choice of $\tau$ from Lemma~\ref{big} or
 Lemma~\ref{big01}).
\end{theorem}
\begin{proof}

    Let $g(x)\coloneqq(h_p^1(x), h_p^2(x), \ldots, h_p^k(x))$ be a hash function
 defined as a concatenation of $k$ random LSH functions presented in Section
 \ref{lsh_sec}.  We introduce the clustering $m : g(\mathbb{R}^d) \rightarrow
 2^{n}$, where each cluster is assigned to the corresponding hash value.  For
 each hash value $\alpha$, the corresponding cluster $m(\alpha)$ is $\{x: g(x) =
 \alpha\}$.

 Since we consider two hashes to be equal when they differ at most by one (see
 Observation~\ref{small}), for hash
 $\alpha$, we need to store the reference for every point that satisfies
 $\dist{\alpha - x} \le 1$.
 The number of such clusters is $3^k$, because the result of each hash function can vary
 by one of $\{-1, 0, 1\}$ and the number of hash functions is $k$. Thus, the memory usage is $\Oc(n
 3^k)$ (see Figure~\ref{fig:schema-picture}).

 \input{schema_picture.tex}

 To preprocess the data, we need to compute the value of the function $g$ for every
 point in the set and then put its reference into $3^k$ cells. Hence, the preprocessing
 time complexity equals $\Oc(n (k d + 3^k))$.

 Eventually, to answer a query, we need to compute $g(q)$ in time $\Oc(k d)$
 and then for every point in $\dist{g(x) - g(q)}_{\infty} \le 1$ remove
 distant points $\dist{x-q}_p > c R$.
 Hence, we need to look up every false-positive to check whether they are within
 distance $c r$ from the query point. We do that
 in expected time $\Oc(d (|P| + k + n \pfp^k))$, because $n \pfp^k$ is the expected
 number of false positives.
\end{proof}
The number of iterations $k$ can be chosen arbitrarily, so we will choose the
optimal value to minimize the query time. This gives the main result of this paper:
\begin{theorem}\label{th_main}
 For any $c>\tau$ and for large enough $n$, there exists a
 \textit{$c$-approximate} nearest neighbor algorithm without false negatives for
 $l_p$, where $p \in [1,\infty]$:
 \begin{itemize}
     \item Preprocessing time: $\Oc(n (\gamma d\log{n}  + (\frac{n}{d})^\gamma)) = \poly{n}$,
    \item Memory usage: $\Oc(n (\frac{n}{d})^\gamma)$,
    \item Expected query time: $\Oc(d (|P| + \gamma \log(n) + \gamma d))$.
 \end{itemize}
 Where $|P|$ is the size of the result, $\gamma = \frac{\ln{3}}{-\ln{\pfp}}$ and $\pfp$ and $\tau$ are chosen as in Theorem \ref{t1}.
\end{theorem}
\begin{proof}

Denote $a = -\ln{\pfp}$, $b=\ln{3}$ and set $k = \ceil{ \frac{\ln{\frac{n a}{d} }}{a}  }$.

Let us assume that $n$ is large enough so that $k \ge 1$. 
Then using the fact that $x^{1/x}$ is bounded for $x>0$ we have:

\begin{displaymath}
    3^k \le 3 \cdot (3^{\ln{\frac{n a}{d}}})^{1/a}   
  = 3 \cdot (\frac{n a}{d})^{b/a}  
  = \Oc((\frac{n}{d})^{b/a})
  = \Oc((\frac{n}{d})^{\gamma}),
\end{displaymath}

\begin{displaymath}
    n \pfp^k 
  = n e^{-a k} \le n e^{-a \frac{\ln(\frac{n a}{d})}{a}}
  = \frac{d}{a} 
  = \Oc(d \gamma),
\end{displaymath}

\begin{displaymath}
    k = \Oc(\gamma \log(n)).
\end{displaymath}

Substituting these values in the Theorem \ref{t1} gives needed complexity guaranties.
\end{proof}

There are two variants of Theorems \ref{t1}, \ref{th_main} and \ref{th_lp}.  In
the first variant, we show complexity bounds for very general class of hashing
functions introduced in Lemma \ref{big}. In the second one, we show slightly better guaranties for hashing
functions which are generated using discrete probability distribution on $\{0,1\}$ introduced in Lemma \ref{big01}.
For simplicity the following discussion is restricted only to the second variant which gives better complexity guaranties. 
The definitions of constants $\pfp_2$ and $\tau_2$ used in this discussion are taken from Lemma \ref{big01} .
For a general case, i.e., $\pfp_1$ and $\tau_1$ taken from Lemma~\ref{big}, we get only slightly worse results.

The complexity bounds introduced in the Theorem \ref{th_main} can be simplified
using the fact that $\ln(x) < x-1$. Namely, we have:

\[\gamma  = \frac{\ln{3}}{-\ln{\pfp_2}} 
          = \frac{\ln{3}}{-\ln(1-\frac{(1-\frac{\tau_2}{c})^2}{2})} 
          < \frac{2\ln{3}}{(1-\frac{\tau_2}{c})^2}.\]  

However, the
preprocessing time is polynomial in $n$ for any constant $c$, it strongly
depends on the bound for probability $\pfp_2$ and $c$. Particularly when $c$ is
getting close to $\tau_2$, the exponent of the preprocessing time might be
arbitrarily large.

To the best of our knowledge, this is the first algorithm that will ensure
that no false negatives will be returned by the nearest neighbor
approximated search and does not depend exponentially on the number of dimensions.
Note that for given $c$, the parameter $\gamma$ is fixed. By
Lemma~\ref{big01}, we have: $\pfp_2 = 1-\frac{(1-\frac{\tau_2^2}{c^2})^2}{2}$, so:

\begin{displaymath}
    \lim_{c \to \infty} {\gamma} =\lim_{c \to \infty} \frac{\ln{3}}{-\ln{\pfp_2}} = \log_2{3} \approx 1.58
    .
\end{displaymath}

If we omit terms polynomial in $d$, the preprocessing time of the algorithm
from Theorem \ref{th_main} converges to $\Oc(n^{2.58})$ ($\Oc(n^{3.71})$ for
general case - see Appendix \ref{agen_pre}.).

\subsection{Light preprocessing}

Although the preprocessing time $\Oc(n^{2.58})$ may be reasonable when there are
multiple, distinct queries and the data set does not change (e.g., static databases,
pre-trained classification, geographical map). Still, unless the number of points is
small, this algorithm does not apply. Here, we will present
an algorithm with a light preprocessing time $\Oc(d n \log{n})$ and $\Oc(n \log{n})$
memory usage where the expected query time is $o(n)$.

The algorithm with light preprocessing is very similar to the algorithm described in
Theorem~\ref{t1}, but instead of storing references to the point in all $3^k$
buckets during preprocessing, this time searching for every point $x$ that matches
$\dist{x-q}_{\infty} \le 1$ is done during the query time.

The expected query time with respect to $k$ is $\Oc(d (|P| +
k + n \pfp^k ) + 3^k)$. During the preprocessing phase we only need to compute
$k$ hash values for each of $n$ points and store them in memory.
Hence, preprocessing requires $\Oc(kdn)$ time and uses $\Oc(nk)$ memory.

\begin{theorem}\label{th_lp}
 For any $c>\tau$ and for large enough $n$, there exists a
 \textit{$c$-approximate} nearest neighbor algorithm without false negatives for
 $l_p$, where $p \in [1,\infty]$:
 \begin{itemize}
     \item Preprocessing time: $\Oc(n d\log{n})$,
     \item Memory usage: $\Oc(n\log{n} )$,
     \item Expected query time: $\Oc(d (|P| +  n^{\frac{b}{a+b}}(\frac{b}{a})^{\frac{a}{b+a}}))$.
 \end{itemize}
 Where $|P|$ is the size of the result, $a = -\ln{\pfp}$, $b=\ln{3}$,  $\pfp$ and $\tau$ are chosen as in Theorem \ref{t1}.
\end{theorem}
\begin{proof}
We set the number of iterations $k = \ceil{ \frac{\ln{\frac{n a}{b}}}{a+b} }$.
Assume $n$ needs to be large enough
so that $k \ge 1$. Since $a$ is upper bounded for both choices of $\pfp$: 
\begin{displaymath}
    3^k \le 3 \cdot 3^{\frac{\ln(\frac{n a}{b})}{a+b}}= 3
    (\frac{n a}{b})^{\frac{b}{a+b}}
    =\Oc(n^{\frac{b}{a+b}})
    .
\end{displaymath}
Analogously:
\begin{displaymath}
    n \pfp^k = n (e^{-a})^k \le
    n e^{-a \frac{\ln(\frac{n a}{b})}{a+b}}=
    n \cdot \Big(\frac{b}{a}\Big)^{\frac{a}{a+b}} \cdot \Big(\frac{1}{n}\Big)^{\frac{a}{a+b}}=
    n^\frac{b}{a + b}\Big(\frac {b}{a}\Big)^{\frac{a}{a + b}}
    .
\end{displaymath}
Hence, for this choice of $k$ we obtain the expected query time is equal to:
\begin{align*}
    \Oc(d (|P| + k + n\pfp^k)) + 3^k &=\Oc(d(|P| + \log{n} + n^\frac{b}{a + b}\Big(\frac {b}{a}\Big)^{\frac{a}{a + b}}) + n^{\frac{b}{a+b}}) \\
    &=\Oc(d(|P| + n^\frac{b}{a + b}\Big(\frac {b}{a}\Big)^{\frac{a}{a + b}}).
\end{align*}
Substituting $k$, we obtain formulas for preprocessing time and memory usage.

\end{proof}

Eventually, exactly as previously for a general distribution from
Lemma~\ref{big}, when $c \to \infty$ we have: $a \to \ln{\frac{3}{2}}$ (see
Theorem \ref{th_lp} for the definition of constant $a$). Hence, for a general
distribution we have a bound for complexity equal to $\Oc(n^{\log_{4.5}{3}})
\approx \Oc(n^{0.73})$. For the discrete distribution from Lemma~\ref{big01},
the constant $a$ converges to $\ln{2}$. Hence, the expected query time
converges to $\Oc(n^{0.61})$.

%% file: schema_picture.tex
\begin{figure*}[ht!]

\centering

\newcommand*{\xMin}{0}%
\newcommand*{\xMax}{10}%
\newcommand*{\yMin}{0}%
\newcommand*{\yMax}{6}%

\begin{tikzpicture}[scale = 0.5]

    \node[above] at (5,7) {$k$};

    \draw [decorate,decoration={brace,amplitude=10pt},rotate=270] (-6.2,0) --
    (-6.2,10.0);

    \foreach \x/\y in {0/4,1/2,2/1,3/3,4/3,5/2,6/4,7/2,8/3,9/1} {
      \fill[blue!10!white] (\x,\y-1) rectangle (\x+1,\y+2);
      \fill[blue]  (\x+0.5,\y+0.5) circle [radius=2.0pt];
     }
    \foreach \x/\y in {0/3,1/3,2/0,3/2,4/4,5/3,6/3,7/3,8/2,9/2} {
        \fill[ForestGreen]  (\x+0.5,\y+0.5) circle [radius=2.0pt];
     }
    \foreach \x/\y in {5/0} {
        \fill[red]  (\x+0.5,\y+0.5) circle [radius=2.0pt];
     }

    \foreach \i in {\xMin,...,\xMax} {
        \draw [very thin,gray] (\i,\yMin) -- (\i,\yMax)  node
        [below] at (\i,\yMin) {};
    }
    \foreach \i in {\yMin,...,\yMax} {
        \draw [very thin,gray] (\xMin,\i) --
        (\xMax,\i) node [left] at (\xMin,\i) {};
    }

\end{tikzpicture}

\caption{Blue dots represent value of $g(q)$ for query. Green dots are
always distant by $1$, hence green and blue points are considered close. At
least one red dot is distant from blue dot by more than $1$, hence red dots will not be
considered close to blue. Thus, algorithm needs to check $3^k$ various possibilities.}
\label{fig:schema-picture}

\end{figure*}
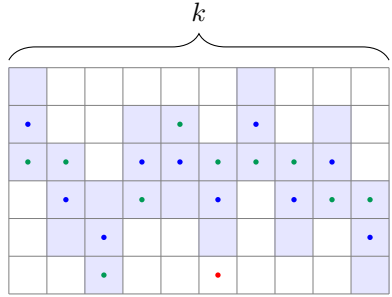

%% file: conclusion.tex
\section{Conclusion and Future Work}

We have presented the $c$-approximate nearest neighbor algorithm without false
negatives in $l_p$ for all $p \in [1,\infty]$ and $c > \max\{\sqrt{d},
d^{1-1/p}\}$. Due to this inequality our algorithm can be used cognately to
the original LSH~\cite{motwani} but with additional guarantees about very close
points (one can set $R' = \sqrt{d} R$ and be certain that all points within
distance $R$ will be returned). In contrast to the original LSH, our algorithm
does not require any additional parameter tunning.

The future work concerns relaxing restriction on the approximation factor $c$
and reducing time complexity of the algorithm or proving that these
restrictions are essential. We wish to match the time complexities given by
\cite{motwani} or show that achieved bounds are optimal. We show the
tightness of our construction, hence to break the bound of
$\sqrt{d}$, one would need to introduce a new technique.

\section{Acknowledgments}

This work was supported by ERC PoC project PAAl-POC 680912 and FET project MULTIPLEX
317532. We would also like to thank Rafa{\l} Lata{\l}a for meaningful discussions.

%% file: appendix.tex
\appendix

\section{Proof of Observation \ref{char}}\label{achar}
\begin{proof}
We will use, the fact that for any $x,y \in \Rspace$ we have
    $\abs{\floor{x} - \floor{y}}
        \le 1 \Rightarrow |x-y|
     < 2$.
Then the following implications hold:
\begin{eqnarray*}
\abs{h_p(x) - h_p(y)} \le 1
\iff &
\biggabs{ \floor{ \frac{\dotpr{x}{v}}{\rhop r} } -
          \floor{ \frac{\dotpr{y}{v}}{\rhop r} } } \le 1
& \implies
\Bigabs{ \frac{\dotpr{x}{y}}{\rhop r} -
         \frac{\dotpr{y}{v}}{\rhop r}  } < 2
\iff
         \\
\iff &
\abs{\dotpr{x-y}{v}} < 2 \rhop r
.
&
\end{eqnarray*}
So, based on the increasing property of the probability:
\begin{displaymath}
    \mathrm{if} \; A \subset B \; \mathrm{then} \; \prob{A} \le \prob{B}
    ,
\end{displaymath}
the inequality of the probabilities holds.

\end{proof}

\section{Proof of Observation \ref{char2}}\label{achar2}

\begin{proof}
We will use the fact that for $x,y \in \Rspace :
\abs{x-y} < 1 \Rightarrow \abs{\floor{x} - \floor{y}} \le 1$).

\begin{eqnarray*}
    \Bigabs{\dotpr{x-y}{v}} < \rhop r
    \iff &
    \Bigabs{ \frac{\dotpr{x}{v}}{\rhop r} - \frac{\dotpr{x}{v}}{\rhop r} } < 1
    & \implies
    \biggabs{ \floor{\frac{\dotpr{x}{v}}{\rhop r}} -
        \floor{\frac{\dotpr{x}{v}}{\rhop r} }} \le 1
    \iff  \\
    \iff &
    \abs{h_p(x) - h_p(y)} \le 1
    &
\end{eqnarray*}

\end{proof}

\section{Proof of Observation \ref{observation-max}}\label{amax}
\begin{proof}
    For every $0 < b \le a$ vectors in $\Rdspace$ satisfy the inequality:
    \begin{equation}\label{ineq-ab}
        \dist{z}_a \le \dist{z}_b \le d^{(\frac{1}{b} - \frac{1}{a})} \dist{z}_a
        .
    \end{equation}

    For $p>2$ we have $\max \{ d^\frac{1}{2} , d^{1-\frac{1}{p}} \} =
    d^{1-\frac{1}{p}}$.  Then, using
    ineqaulity~\eqref{ineq-ab} for $a=p$ and $b=2$ we have:

    \begin{displaymath}
        \dist{z}_2 \ge \dist{z}_p = \frac{\rhop}{d^{1-\frac{1}{p}}} \dist{z}_p
        =
        \frac{\rhop}{\max \{d^\frac{1}{2}, d^{1 - \frac{1}{p} } \} } \dist{z}_p
    \end{displaymath}

    For $1 \le p \le 2$ we have $\max \{ d^\frac{1}{2},
    d^{1-\frac{1}{p}} \} = d^\frac{1}{2}$. Analogously by using
    inequality~\eqref{ineq-ab} for $a = 2$ and $b=p$:

    \begin{displaymath}
        \dist{z}_p \le d^{\frac{1}{p} - \frac{1}{2}} \dist{z}_2 =
        \dist{z}_2 \frac{d^{\frac{1}{2}}}{\rhop}
    \end{displaymath}

    Hence, by dividing both sides we have:

    \begin{displaymath}
        \dist{z}_p \frac{\rhop}{\max \{ d^\frac{1}{2}, d^{1-\frac{1}{p}} \}} \le
        \dist{z}_2
    \end{displaymath}

\end{proof}

\section{Hoeffding bound} \label{hoef}

Here we are going to show all technical details used in the proof in the Section~\ref{tight}.
Let us start with the Hoeffding inequality.
Let $X_1, \ldots, X_d$ be bounded independent random variables: $a_i \le X_i \le b_i$
and $\overline X$ be the mean of these variables
$ \overline X = \sum_{i=1}^{d}X_i / d $.
Theorem 2 of Hoeffding~\cite{Hoeffding:1963} states:

\begin{align*}
    \prob{\abs{\overline X - \Ex\left[ \overline X \right]}\geq t} &
    \leq 2 \cdot \Exp{-\frac{2d^2t^2}{\sum_{i=1}^d(b_i - a_i)^2}}
    .
\end{align*}

In our case, $D_1, \ldots, D_d$ are bounded by $a_i = -1 \le D_i \le 1 = b_i$ with $\Ex D_i=0$. Hoeffding
inequality implies:

\begin{align*}
    \prob{\left |\frac{\sum_{i=1}^{d} D_i}{d}  \right | \geq t} &
    \leq 2 \cdot \Exp{-\frac{2d^2t^2}{\sum_{i=1}^d(b_i - a_i)^2}}
    = 2 \cdot \Exp{-\frac{dt^2}{2}}
    .
\end{align*}

Taking $t=d^{-1/2 +\epsilon}$ we get the claim:

\begin{align*}
    \prob{\biggabs{\frac{\sum_{i=1}^{d} D_i}{d}} \geq d^{-1/2 +\epsilon}} &
    \leq 2 \cdot \Exp{-\frac{d^{2 \epsilon}}{2}}
    .
\end{align*}
\section{Preprocessing complexity bounds for the distributions introduced in Lemma \ref{big}}\label{agen_pre}
By Lemma~\ref{big}, we have: $\pfp_1 = 1-\frac{(1-\frac{\tau_1^2}{c^2})^2}{3}$, so:
\begin{displaymath}
    \lim_{c \to \infty} {\gamma}=\lim_{c \to \infty} \frac{\ln{3}}{-\ln{\pfp_1}} = {\frac{\ln{3}}{\ln{1.5}}}
        \approx {2.71}
    .
\end{displaymath}

If we omit terms polynomial in $d$, the preprocessing time of the algorithm
from Theorem \ref{th_main} converges to $\Oc(n^{3.71})$.